\documentclass[review]{elsarticle}

\usepackage{lineno,hyperref}
\modulolinenumbers[5]

\journal{Journal of Information Sciences}

\newcommand{\ignore}[1]{}

\usepackage{amsmath,amssymb,amsfonts,amsthm}
\usepackage{adjustbox}
\usepackage{graphicx}
\usepackage{textcomp}
\usepackage{xcolor}

\usepackage{eso-pic}
\usepackage{float}

\usepackage{mathrsfs}
\usepackage{algorithm}
\usepackage[noend]{algpseudocode}
\usepackage{booktabs}
\usepackage{setspace}

\usepackage{array, makecell} 
\usepackage{tikz}
\usepackage{multirow}
\usepackage{epstopdf}

\usepackage{caption}
\usepackage{footnote}

\usepackage{cancel}
\usepackage{pifont}
\usepackage{natbib}

\makesavenoteenv{tabular}

\newcommand\With{\textbf{with}}

\newcommand\From{\textbf{from}}
\newcommand\Broadcast{\textbf{broadcast}}

\theoremstyle{definition}

\newtheorem{theorem}{Theorem}[section]

\newtheorem{lemma}[theorem]{Lemma}

\newtheorem{definition}[theorem]{Definition}

\newtheorem{observation}[theorem]{Observation}

\newtheorem{remark}[theorem]{Remark}

\title{Marvel DC: A Blockchain-Based Decentralized and Incentive-Compatible Distributed Computing Protocol\tnoteref{t1}}
\tnotetext[t1]{This document is part of a project that has received funding from the European Union's Horizon 2020 research and innovation programme under grant agreement number 814284}

\newcommand{\adversary}{\mathcal{A}}
\newcommand{\adversarialBound}{\alpha}
\newcommand{\avgRepChange }{\textit{avgRepChange} }

\newcommand{\badComputers}{\indicesForComputation_{bad}}

\newcommand{\baseRep}{\textit{baseRep}}

\newcommand{\blockchainHeight}{H}
\newcommand{\blockProposer}{blockProposer}

\newcommand{\compDecKey}{\textit{compDecKey}}

\newcommand{\compEncKey}{\textit{compEncKey}}

\newcommand{\computation}{\textit{calc}}
\newcommand{\computationCost}{\textit{cost}(\computation)}
\newcommand{\computationStart}{\textit{start}}
\newcommand{\computationProposers}{\textit{proposers}}
\newcommand{\computationResult}{\textit{result}}
\newcommand{\computationReward}{\textit{Reward}(\computation)}
\newcommand{\compute}{\textit{compute}}
\newcommand{\computer}{\textit{c}}
\newcommand{\computers}{\textit{C}}

\newcommand{\contract}{\textit{contract}}

\newcommand{\counter}{\textit{ctr}}
\newcommand{\currentHeight}{\blockchain.\textit{height}}
\newcommand{\denominator}{\textit{denominator}}

\newcommand{\escrowComputer}{\escrow_\textit{comp}}
\newcommand{\escrowRequester}{\escrow_\textit{req}}
\newcommand{\expectedBaseRepChange}{E_\textit{repChange}}

\newcommand{\finalizeDeadline}{\textit{FD}}

\newcommand{\genProbSelect}{\textit{genProbSelect}}

\newcommand{\getDistance}{\textit{dist}}

\newcommand{\getEscrows}{\textit{getEscrows}}
\newcommand{\getFinalizeDeadline}{\textit{getFinalizeDeadline}}

\newcommand{\getLength}{\textit{len}}

\newcommand{\getMinimum}{\textit{min}}
\newcommand{\getRandomness}{\textit{genRandom}}
\newcommand{\goodComputers}{\indicesForComputation_{good}}
\newcommand{\goodResponses}{\responses_{good}}
\newcommand{\hash}{\textit{SHA}}
\newcommand{\isValid}{\textit{valid}}
\newcommand{\indicesForComputation}{I}
\newcommand{\initRep}{\textit{initRep}}

\newcommand{\marginOfError}{\sigma}
\newcommand{\minRep}{\textit{minRep}}
\newcommand{\numComputers}{{n_{comp}}}
\newcommand{\numComputersBound}{{n_{\secParam}}}
\newcommand{\numPlayers}{n}
\newcommand{\numToReward}{{n_{reward}}}

\newcommand{\playerReward}{\reward_i}
\newcommand{\playerSetDescription}{\{\player_1,...,\player_\numPlayers\}}

\newcommand{\probBadRating}{\gamma}
\newcommand{\probGoodRating}{\omega}
\newcommand{\probSelect}{\textit{probSelect}}

\newcommand{\protocolName}{\text{Marvel DC}}

\newcommand{\randomSeed}{\textit{randomSeed}}
\newcommand{\rate}{\textit{rate}}
\newcommand{\rateComputations}{\textit{rateComputations}}

\newcommand{\reputations}{\textit{Reps}}
\newcommand{\requester}{\textit{requester}}

\newcommand{\response}{\textit{response}}
\newcommand{\responses}{\textit{responses}}

\newcommand{\result}{\textit{result}}
\newcommand{\results}{\textit{results}}
\newcommand{\reward}{\textit{Reward}}

\newcommand{\selectComputers}{\textit{selectComputers}}
\newcommand{\selectionProb}{\textit{probSelect}}

\newcommand{\step}{\textit{step}}
\newcommand{\stepComputing}{\textit{computing}}
\newcommand{\stepFinalized}{\textit{finalized}}
\newcommand{\sumReps}{\textit{sumReps}}
\newcommand{\tagFinalize}{\textit{FINALIZE}}
\newcommand{\tagRegister}{\textit{REGISTER}}
\newcommand{\tagRequest}{\textit{REQUEST}}
\newcommand{\tagResponse}{\textit{RESPONSE}}

\newcommand{\target}{\tau}
\newcommand{\targetFunction}{\textit{f}_\target}
\newcommand{\targetFunctions}{\textit{tFunctions}}

\newcommand{\txFee}{\textit{Fee}(\tx)}

\newcommand{\updateReputations}{\textit{updateReputations}}

\newcommand{\balance}{\textit{bal}}

\def\bitcoin{%
  \leavevmode
  \vtop{\offinterlineskip 
    \setbox0=\hbox{B}%
    \setbox2=\hbox to\wd0{\hfil\hskip-.03em
    \vrule height .3ex width .15ex\hskip .08em
    \vrule height .3ex width .15ex\hfil}
    \vbox{\copy2\box0}\box2}}

\newcommand{\blockchain}{\textit{Blockchain}}

\newcommand{\commit}{f_\textit{com}}
\newcommand{\commitment}{\textit{com}}
\newcommand{\commitments}{\textit{Com}}
\newcommand{\commRandomness}{\textit{r}}
\newcommand{\commSerialNum}{\textit{S}}

\newcommand{\cryptoParam}{\kappa}

\newcommand{\escrow}{\textit{escrow}}

\newcommand{\feeRelayer}{\textit{Fee}(\textit{relay})}
\newcommand{\getAppend}{\textit{.append}}
\newcommand{\getBalance}{\textit{.balance}}

\newcommand{\getTransfer}{\textit{.transfer}}

\newcommand{\player}{\textit{P}}

\newcommand{\playeri}{\player_i}

\newcommand{\proofZK}{\pi}

\newcommand{\regToken}{\textit{regID}}

\newcommand{\revealTXDelay}{T}
\newcommand{\secParam}{\psi}

\newcommand{\tx}{\textit{tx}}

\newcommand{\yes}{\textcolor{green}{\ding{51}}}
\newcommand{\no}{\textcolor{red}{\ding{55}}}

\newcommand{\decrypt}{decrypt}

\newcommand{\encrypt}{encrypt}

\newcommand{\ourBAR}{\text{ByRa}}

\algdef{SE}[EVENT]{Upon}{EndUpon}[1]{\textbf{upon}\ #1\ \algorithmicdo}{\algorithmicend}\algtext*{EndUpon}
\newcommand{\newconstruct}[5]{%
  \newenvironment{ALC@\string#1}{\begin{ALC@g}}{\end{ALC@g}}
   \newcommand{#1}[2][default]{\ALC@it#2\ ##2\ #3%
     \ALC@com{##1}\begin{ALC@\string#1}}
   \ifthenelse{\boolean{ALC@noend}}{
     \newcommand{#4}{\end{ALC@\string#1}}
   }{
     \newcommand{#4}{\end{ALC@\string#1}\ALC@it#5}
   } 
}

\newcommand{\assign}{\leftarrow}

\newcommand{\ident}[1]{\mathit{#1}}
\def\newident#1{\expandafter\def\csname #1\endcsname{\ident{#1}}}

\newcommand{\SPACE}{\vspace{3mm}}

\newcommand{\logicalAnd}{\textbf{and}}
\newcommand{\logicalOr}{\textbf{or}}
\newcommand{\logicalNot}{\textbf{not}}

\begin{document}

\author[add1,add2]{Conor McMenamin}
\ead{conor.mcmenamin@upf.edu}

\author[add1,add3]{Vanesa Daza}
\ead{vanesa.daza@upf.edu}

\address[add1]{Department of Information and Communication Technologies, Universitat Pompeu Fabra, Barcelona, Spain}
\address[add2]{NOKIA Bell Labs, Nozay, France}
\address[add3]{CYBERCAT - Center for Cybersecurity Research of Catalonia}

\maketitle

\section{Abstract}

Decentralized computation outsourcing should allow anyone to access the large amounts of computational power that exists in the Internet of Things. Unfortunately, when trusted third parties are removed to achieve this decentralization, ensuring an outsourced computation is performed correctly remains a significant challenge. In this paper, we provide a solution to this problem. 

We outline $\protocolName$, a fully decentralized blockchain-based distributed-computing protocol which formally guarantees that computers are strictly incentivized to correctly perform requested computations. Furthermore, $\protocolName$ utilizes a reputation management protocol to ensure that, for any minority of computers not performing calculations correctly, these computers are identified and selected for computations with diminishing probability. We then outline Privacy $\protocolName$, a privacy-enhanced version of $\protocolName$ which decouples results from the computers which computed them, making the protocol suitable for computations such as Federated Learning, where results can reveal sensitive information about that computer that computed them. 
We provide an implementation of $\protocolName$ and analyses of both protocols, demonstrating that they are not only the first protocols to provide the aforementioned formal guarantees, but are also practical, competitive with prior attempts in the field, and ready to deploy. 

\textbf{Keywords:} Distributed Computing, Decentralization, Blockchain, Incentives

\section{Introduction}

The distributing of computation among computers has beckoned a never-before-seen level of computing power available to average users with access to as little as a smart phone and an average internet connection. Distributed computations (DCs) have typically been outsourced directly to one of the few centralized Big Tech. providers such as Amazon Web Services, or Google Cloud. Unfortunately, centralized services like these have many drawbacks. Monopoly of resources, centralized trust, restricted access for certain clients, and in the case of Federated Learning (FL), lack of diverse data-sets make these centralized services unfit for many users and purposes. Full decentralization solves all of these issues. Unfortunately, decentralizing distributed-computing protocols brings many new challenges that are largely protected against in the centralized setting. A significant issue in many  existing decentralized DC implementations \cite{BlockchainDecentralisedFedLearningLi,FLChainBao,FLChainMajeed,FedLearningMeetsBlockchainMa,BlockFLKim,DeepChainWeng, BlockchainRepAwareFedLearningRehman,IncentiveMechFLBlockchainToyoda} is the accurate rewarding of computers to incentivize rational computers, computers who try to maximize their utility (e.g. in cryptocurrency tokens), to correctly perform outsourced computations. In decentralized settings such as blockchain protocols, players only follow an action(s) if that action(s) is(are) strong incentive compatible, resulting in strictly higher payoffs than the alternatives. One method to combat this is, for a given computation, to use zero-knowledge (ZK) tools to prove that a computer performed the actions as prescribed by the requester \cite{FairnessIntegrityPrivacyBlockchainBasedFLSystemRuckel}. Although theoretically any computation can be encoded in this way, the practicality of requiring requesters, typically with low computing power by the nature of outsourcing, to encode their computation as a ZK-circuit to allow for the proving of correct computation is an open question. Furthermore, the pre-computation work required by requesters in \cite{FairnessIntegrityPrivacyBlockchainBasedFLSystemRuckel} is itself intensive (3-40 minutes using 16 virtual CPU cores, and 64 GB RAM), and not appropriate for computationally-restricted requesters, defeating a significant purpose of DC protocols. In this paper, we address all of these issues.

\subsection{Our Contribution}

We present $\protocolName$, a generic (can be adapted to any computational problem outputting results in Euclidean space) blockchain-based decentralized DC protocol which addresses the many gaps that exist in outsourcing computations without the use of a trusted third-party (TTP). Namely, $\protocolName$ provides for the first time in literature a decentralized DC protocol which ensures rational computers are strongly incentivized to follow the protocol. This is a significant advancement in a field where existing claims of incentive compatibility do not account for computers trying to maximize their tokens, with no viable solution, to the best of our knowledge, for the distributing of tokenized rewards in a distributed and decentralized manner. 

Furthermore, $\protocolName$ utilizes reputations to isolate correctly-performing computers when selecting computers for computations. This allows $\protocolName$ to efficiently remove adversarial computers from the protocol. As these reputations are maintained on the blockchain itself, through careful construction (Section \ref{sec:Reputation}) this reputation protocol neither affects the decentralisation or strong incentive compatibility of the protocol. When using a fixed number of computers per computation, our protocol shares all of the benefits of \cite{CoUtileP2PDecentralizedComputing}.
Our description of $\protocolName$ can be adapted to run on any smart-contract enabled blockchain, and as such, can make use of the vast existing communities which exist on such blockchains. This is a further improvement on protocols which require the recruitment and constant participation of an independent network of computers. In a blockchain, where computers form a subset of users, computers can be dormant until required to perform a computation. By deploying on an existing blockchain, any player in that blockchain can also participate in $\protocolName$ as a computer and/or requester. We summarize the main contributions of $\protocolName$ as follows:

\begin{itemize}
    
    \item \textbf{Strong incentive compatibility in expectation}. In Section \ref{sec:SINCE}, we describe how to program rewards such that it is strong incentive compatible in expectation for every rational player (requesters, computers and/or block-producers) in the blockchain system to follow the protocol.
    
    \item \textbf{Handling of symmetric/asymmetric utilities}. By tokenising the protocol rewards, we are able to handle both symmetric (rational computers and requesters only want to produce good results) and asymmetric (rational computers want to be compensated financially) utilities. 
    
    \item \textbf{Decentralization}. Encoding $\protocolName$ in a generic manner for any tokenized smart-contract enabled blockchain, all players in such a  blockchain system are able to participate in $\protocolName$. Furthermore, all of the tools used in $\protocolName$ are already in use in smart-contract enabled decentralized systems such as Ethereum \footnote{\url{https://ethereum.org/en/}} and Harmony \footnote{\url{https://www.harmony.one/}} (our smart-contract encoding \cite{MarvelDCGithub} is written in Solidity \cite{MarvelDCGithub}, a programming language interpretable by both blockchains ), demonstrating our encoding can utilize the decentralisation of such blockchain ecosystems.

\end{itemize}

On top of these main contributions, $\protocolName$ also retains additional desirable properties that should be present in any DC protocol:

\begin{itemize}
    
    \item \textbf{Fair selection}. Our protocol uses a decentralized pseudorandom function (such as a verifiable on-chain oracles\footnote{\url{https://docs.chain.link/docs/chainlink-vrf/}}, or those used in  \cite{ShieldedComputationsInSmartContractsOvercomingForksVisconti}) to provide a random seed used to select computers based on reputation. This removes the possibility of reputation-based manipulation by clusters of colluding computers who may try to prioritize the selection of computers in the colluding cluster, as is allowed in existing works \cite{CoUtileP2PDecentralizedComputing,CoUtileReputationProtocol, FairnessIntegrityPrivacyBlockchainBasedFLSystemRuckel}. 
    
    \item \textbf{Proper management of new peers}. All new computers entering $\protocolName$ are given the same initial reputation, which eventually becomes significantly lower than the average reputation of active honestly-behaving computers.
    
    \item \textbf{Low communication complexity $\&$ cost}. By utilising a blockchain, computers participating in $\protocolName$ are only required to monitor the blockchain for computation requests, and respond when they themselves are selected to participate in the computation. Furthermore, by utilising cost-effective blockchains such as Harmony, the costs for participation in $\protocolName$ are fractions of a US dollar, as shown in Section \ref{sec:Analysis}.
    
\end{itemize}

Given these properties of $\protocolName$, we then outline a series of privacy-enhancing amendments, culminating in Privacy $\protocolName$, a protocol in which results cannot be linked to the computer which provided the result, except with significant additional cost to the player who causes the revelation. Privacy $\protocolName$ retains all of the decentralisation and incentive compatibility guarantees of $\protocolName$, while also adding a layer of privacy which makes it appropriate for sensitive computations where knowing a certain computer computed a particular result can be used to infer private/unwanted information about the computer. This is particularly interesting in the case of FL, where computers are asked to use private data set. 

To do this, we make use of existing ZK set-membership tools (described in Section \ref{sec:ZK}). These tools are used in protocols such as Zerocoin \cite{ZerocoinGreen} and Tornado Cash\footnote{\url{https://tornado.cash/}}, with the latter of which currently deployed and in use on the Ethereum blockchain. Computers encrypt results, and publish them anonymously to the blockchain through a network of relayers: blockchain participants who are incentivized to submit transactions on behalf of other players (explained in detail in Section \ref{sec:Relayers}). The set of good and bad results are published as in $\protocolName$, and as results are not attached to the computer that sent them, computers prove membership among the set of computers submitting good/bad results using the aforementioned ZK set-membership tools. 

Although our privacy techniques are not novel outside of DC, and in fact reduce privacy compared to DC protocols like \cite{FairnessIntegrityPrivacyBlockchainBasedFLSystemRuckel}, the combination of decentralisation, proven strong incentive compatibility and the ability to apply one smart-contract instance of Privacy $\protocolName$ to any computational problem  with output in Euclidean space (summarized in Table \ref{table:RW}) stands as an additional novel contribution.

\section{Related Work}\label{sec:RW}

\begin{table}[H]
\begin{center}
\scalebox{0.7}{
  \begin{tabular}{ |l|c|c|c|c|}
    \toprule
    Protocol & \makecell{Tokenized \\ Rewards} & \makecell{Strong Incentive \\ Compatible} &  Computation-Independent & \makecell{Diminishing Adversarial \\ Selection Prob.}  \\
    \bottomrule
    \cite{BlockchainDecentralisedFedLearningLi,FLChainBao,FLChainMajeed,FedLearningMeetsBlockchainMa,BlockFLKim,DeepChainWeng, BlockchainRepAwareFedLearningRehman} & \no & \no   & \yes & \yes / \no  \\
    \midrule
    Toyoda et al. \cite{IncentiveMechFLBlockchainToyoda}  & \yes & \no   & \yes & \no  \\
    \midrule
    Ruckel et al. \cite{FairnessIntegrityPrivacyBlockchainBasedFLSystemRuckel}  & \yes & \yes   & \no * & \no \\
    \midrule
    Marvel DC & \yes & \yes  & \yes  & \yes \\
    \bottomrule
  \end{tabular}}
  \caption{Comparison of incentive-aware FL protocol designs. *Computation-independence refers to the ability of a particular smart-contract encoding to be re-used for many computations. The ZK circuit-encodings of \cite{FairnessIntegrityPrivacyBlockchainBasedFLSystemRuckel} must be generated anew for each type of computation, placing significant upfront costs on computation requesters. } \label{table:RW}
\end{center}
\end{table}

There are many intertwined areas of research regarding the decentralized outsourcing of computations to distributed sets of potentially untrusted peers. Strong advancements have been made with respect to single computer outsourcing, with \cite{BlockchainBasedDistributedComputingPaymentsZhang,OBFPLin,OptiSwapEckey} providing variations of such pairwise protocols.

These 2-player protocols involve several rounds of communication between the requester and computer. However, none discuss the problem of computer selection or rewarding in the presence of many competing computers. This limits their scope for outsourcing computations where the participation of many computers is required, such as in FL. 

Blockchain-based DC protocols  \cite{BlockchainDecentralisedFedLearningLi,FLChainBao,FLChainMajeed,FedLearningMeetsBlockchainMa,BlockFLKim,DeepChainWeng, BlockchainRepAwareFedLearningRehman} have been well-studied also. Unfortunately, all of these papers consider a blockchain or distributed system in which all parties share one utility, that is, the ability to use/benefit from a well-trained shared  model, which gives correct behaviour by definition. Such an assumption makes these protocols and their resulting analyses inappropriate in the presence of untrusted peers and/or asymmetric utilities. 

In \cite{IncentiveMechFLBlockchainToyoda}, a protocol and general framework for incentive mechanism design, within FL protocols where players measure utility tokenomically, are proposed. Computer registration and reward distribution must all be performed by players within the system. Computer registration is performed by an ``administrator", which prevents decentralization and encourages collusion. Furthermore, computation rewards are distributed based on votes of players within the system. The incentive compatibility of this choice is not considered, and is non-trivial to implement. In $\protocolName$, rewards are distributed deterministically as part of the smart-contract execution on requester inputs. We prove that rational requesters always submit messages correctly to the blockchain, and thus, correct rewarding is strong incentive compatible. Moreover, \cite{IncentiveMechFLBlockchainToyoda} has no mechanism to identify Byzantine computers and diminish/remove their ability to participate in the protocol. In $\protocolName$, this is achieved using reputations.

Recently, an extensive survey of existing attempts to construct privacy-preserving FL protocols\cite{privacyPreservingFLSurveyLiu} was published. Of the investigated works, the most promising for achieving an incentive-compatible decentralized protocol is \cite{FairnessIntegrityPrivacyBlockchainBasedFLSystemRuckel}.  In \cite{FairnessIntegrityPrivacyBlockchainBasedFLSystemRuckel}, ZK-proofs are used in the computing stage to prove that a computation was performed correctly. In a decentralized setting however, this raises many challenges, as each type of computation requires its own ZK circuit-generation/trusted set-up to generate the target function. In contrast, $\protocolName$ can be implemented using existing, blockchain-deployed ZK-tools and generalised rewarding functions. Moreover, as acknowledged by the authors of \cite{FairnessIntegrityPrivacyBlockchainBasedFLSystemRuckel}, although their methodology ensures models were trained correctly, it does not guarantee the models were trained on appropriate data. A proposed solution is using ``certified sensors", equivalent to TTPs, a non-viable solution in a decentralized setting. As the rewarding functions in $\protocolName$ reward players based on the relative quality of the results, and not just based on the fact that a series of computations has been performed correctly, independent of the data on which the computations are being performed, we are able to avoid this issue.

\section{Preliminaries}

\subsection{Blockchains}

In this paper, we are interested in a distributed set of $\numPlayers$ players $\playerSetDescription$ interacting with one and other inside a blockchain protocol. These players send and receive stake among one another, along with the functionality to encode programs to run within the blockchain protocol in the form of \textit{smart contracts}. A foundational assumption for the functioning of blockchain protocols is that a majority of players in the protocol act \textit{rationally}, following the protocol if they are incentivized to do so. The non-rational players are known as \textit{Byzantine}, and are modelled as deviating from the protocol either maliciously or randomly, and being controlled by a single adversary. This player model is known as the ByRa model \cite{SMRwithoutHonestPlayers}, which we use in this paper. The ByRa model is a necessary improvement on the legacy BAR model \cite{BARGossip} for the true consideration of incentives in distributed systems, removing any dependencies on altruistic, honest-by-default players which cannot be assumed to exist in incentive-driven protocols like blockchain/DC protocols.   

\begin{spacing}{1.5}
\begin{definition}\label{def:PlayerModel}
    The \textit{$\ourBAR$ model} consists of \textit{Byzantine} and \textit{Rational} players. A player is:
    \begin{itemize}
        \item \textit{Byzantine} if they deviate arbitrarily from a recommended protocol with unknown utility function. Byzantine players are chosen and controlled by an adversary $\adversary$.
        \item \textit{Rational} if they choose the strategy which maximizes their utility assuming all other players are rational.
    \end{itemize}
        
\end{definition}
\end{spacing}

It is the aim of this paper to construct a DC protocol that ensures rational players always follow the protocol, a property known as strong incentive compatibility in expectation. In this paper, rational players utility is measured in blockchain-based tokens. 
Based on similar assumptions to \cite{Hawk}, the blockchain protocol acts conceptually as a public ledger managed by a TTP. In reality, it is the following of the blockchain protocol by some majority of players using the blockchain that replicates this TTP. The protocol provides availability and correctness of the programs being run through the protocol, but does not provide privacy. That is, any player can observe the current state of all programs being run on the blockchain, and can verify that this state has been reached through the correct running of these programs. However, player inputs to these programs must be committed publicly to the blockchain before they can be passed to the smart contract, and as such, it will be an important requirement of designing a protocol involving smart contract interaction, through transactions, that the blockchain will accept these transactions in a timely fashion. In our system, this is achieved using incentivization. 

\ignore{
Smart contracts, like players, have an associated balance of stake. The sending of stake by a smart contract can only be done by satisfying specified logical conditions encoded in the smart contract. Examples of these conditions are the generation of small hash value by a player(Proof-of-work), the calling of the contract by members of a certain set of addresses, such as the contract writer, or the passing of a specified time/number of blocks (Time-lock). We consider smart contracts taking as input specifically-labelled transactions from players, containing payments and arbitrary textual data, and also the current height of the blockchain. The contract then produces some output, distributing, either to other contracts, players, or to burn, some amount of stake up to a maximum of the stake input to the contract. Only stake that is sent to a contract can be distributed by the contract.
}

\subsection{Zero-Knowledge Primitives}\label{sec:ZK}

The aim of this section is to outline existing \textit{non-interactive zero-knowledge} (NIZK) tools for set membership, such as those stemming from papers like \cite{ZerocoinGreen,ZCash,PairingBasedNIZKsGroth,ZKProofsSetMembershipBenarroch,SemaphoreWhitehat}, which are used in our privacy-enhanced DC protocol. We define these tools generically, allowing for the adoption of any secure NIZK set-membership protocol into $\protocolName$, as we only require a common functionality that is shared by all of them.

In the rest of the paper, we let $*$ denote any value. Computers privately generate two bit strings, the \textit{serial number} $\commSerialNum$ and \textit{randomness} $\commRandomness$, with $\commSerialNum, \ \commRandomness \in \{0,1\}^{O(\cryptoParam)}$ for some security parameter $\cryptoParam$. Computers then commit to these values by generating a commitment $\commitment \assign \commit(\commSerialNum,\commRandomness)$ where $\commit(*)$ is some cryptographically-secure commitment function (such as a Pedersen Commitment \cite{PedersenCommitments}). This $\commitment$ is then published to the blockchain, with the set of all commitments denoted by $\commitments$. With this in mind, we now define the key functions that allow us to reason about the NIZK proving of the membership of some commitment $\commitment$ in $\commitments$:

\begin{itemize}
    \item Verify$(\proofZK)$: For a ZK proof $\proofZK$, returns 1 if and only if $\proofZK$ is a valid proof of knowledge.  
    \item MemVerify$(\commitments,\commitment):$ Returns 1 if and only if $\commitment \in \commitments$.
    \item NIZKPoK$\{(\commitment, \commRandomness):$ MemVerify($\commitments, \commitment )=1 \, \And \commitment= \commit(\commSerialNum,\commRandomness) $ $\} \rightarrow \proofZK $: Returns a NIZK proof of knowledge of a commitment $\commitment$ and randomness $\commRandomness$ which satisfies  MemVerify$(\commitments,\commitment)$=1 and $\commitment= \commit(\commSerialNum,\commRandomness)$. The variables $(\commitment, \commRandomness)$ are assumed to be known only by the prover, while all other variables and functions are know by the verifier. Specifically, this function depends on the serial number $\commSerialNum$ being revealed. This revelation identifies to a verifier when a proof has previously been provided for a particular, albeit unknown, commitment as the prover must reproduce $\commSerialNum$. This is used in conjunction with an escrow to enforce the correct participation of computers in our privacy-enhanced DC protocol.
    \item NIZKSoK($m$)$\{(\commitment, \commRandomness):$ MemVerify$(\commitments,\commitment)=1 \, \And \commitment= \commit(\commSerialNum,\commRandomness) $ $\} \rightarrow \proofZK$: for an arbitrary message $m$, this function returns an NIZK proof which proves that the person who chose $m$ can also produce NIZKPoK$\{ $ $(\commitment, $ $ \commRandomness):$ MemVerify$(\commitments,\commitment)=1 \, \And \commitment= \commit(\commSerialNum,\commRandomness) $ $\}$. As such, NIZKSoK($m$) acts as a \textit{signature of knowledge} on $m$, as coined in \cite{ZerocoinGreen}.
\end{itemize}

In this paper, we assume the public NIZK parameters are set-up in a trusted manner\footnote{Such as a Perpetual Powers of Tau ceremony, as used in Zcash \url{https://zkproof.org/2021/06/30/setup-ceremonies/}}. 

\subsection{Relayers}\label{sec:Relayers}

A fundamental requirement for transaction submission in blockchains is the payment of some transaction fee to simultaneously incentivize block producers to include the transaction, and to prevent denial-of-service/spamming attacks. However, this allows for the linking of player transactions, balances, and their associated transaction patterns.
To counteract this, we utilize the concept of \textit{transaction relayers}\footnote{Ox \url{https://0x.org/docs/guides/v3-specification}, Open Gas Station Network \url{https://docs.opengsn.org/}, Rockside \url{https://rockside.io/}, Biconomy \url{https://www.biconomy.io/}}. 

When a computer wishes to submit a transaction anonymously to the block- chain, the computer publishes a proof of membership to the relayer mempool, as well as the desired transaction and a signature-of-knowledge cryptographically binding the membership proof to the transaction, preventing tampering. As the relayer can verify the proof of membership, the relayer can also be sure that if the transaction is sent to certain smart-contracts (such as those used in our DC protocols), the relayer will receive a fee.  As such, a model in which players submit transactions to relayers, who then in turn submit those transactions to the blockchain for a guaranteed fee is a straightforward extension of a model in which players submit transactions directly to the blockchain themselves. Furthermore, if the transaction contains a proof of membership that is NIZK and the message is broadcast anonymously (using the onion routing (Tor) protocol\footnote{\url{https://www.torproject.org/}} for example), the relayer can only infer that the player sending the transaction is a member of a particular set. 

\section{Constructing a Strong Incentive Compatible DC Protocol}\label{sec:SINCE}

For the sake of comprehensibility, we first describe an idealized protocol for the distribution of computation between a set of computers. We then demonstrate how to instantiate such a protocol using existing blockchain technology. In this description, we consider a requester who has a computation $\computation$ whose calculation the requester wishes to outsource to some subset of available computers $\computers$. Furthermore, the requester is aware of a threshold $\numComputersBound$ such that for any random sample of $\numComputersBound$ computers without replacement from $\computers$, a majority of those computers are rational.

We consider the output of all computations in this paper as a point in $l$-dimensional space. To reason about the goodness of a computation result, for every computation we assume the existence of a deterministic function which takes the set of computation responses, and given a majority of correctly computed results, outputs a target value $\target$, which is computed correctly with probability $1- \epsilon$, for some $\epsilon<0.5$. For deterministic calculations, the mode is such a function. In FL where results are model gradients, the Krum function \cite{KrumMLwithAdversariesBlanchard} is such an aggregation function. We let $\getDistance(x,y)$ be the Euclidean distance function. For $x,y \in \mathbb{R}^l$, $\getDistance(x,y)=\sqrt{(x_1-y_1)^2 +...+(x_l-y_l)^2}$.

Consider the set of computation results $\{\result_1,...,\result_\numComputers \}$. We require for any pair ($\result_{+}$, $\result_{-}$), with $\result_{+}$ a correctly computed result and $\result_{-}$ an incorrectly computed result the following holds for any $\marginOfError>0$:

\vspace{-0.5cm}

\begin{equation}\label{eq:BaseCase}
    P( \getDistance(\result_{+}, \target) < \marginOfError) > P( \getDistance(\result_{-}, \target) < \marginOfError).
\end{equation} 
Given this, for any subset of computation results taken in ascending order using the function $\getDistance$, the expected number of these computations being correctly computed is greater than those being incorrectly computed. With respect to deterministic computations, letting $\target$ be the median or mode, all correctly computed results will be distance 0 to $\target$. For non-deterministic computations, it is less clear. The function used to compute $\target$ must be chosen such that Equation \ref{eq:BaseCase} holds.

Now, consider the following DC protocol run by a trusted third party (TTP) who enforces the correct participation of all rational players. The proceeding sections then replace this TTP, in order to achieve full decentralization, through the use of a smart contract-enabled blockchain and a strong incentive compatible protocol ($\protocolName$, described in Section \ref{sec:MarvelAlgos}) to be run therein.

\textbf{Idealized distributed computing protocol}. Requesters repeatedly enter the system with independent functions for computation. A requester wishing to avail of the distributed computation of some function $\computation$ submits $\computation$ to the TTP, as well as some number of computers $\numComputers>\numComputersBound$. The TTP then selects $\numComputers$ from the set of available computers $\computers$. The computers respond to the TTP with their computation of $\computation$, who then sends all of the computation results to the requester. 

\vspace{0.25cm}
In a distributed setting, such TTPs do not exist. Therefore, to enforce the correct participation of rational players we need to utilize a mixture of cryptography and incentivisation. With this in mind, we first describe how to generically construct a reward mechanism such that rational players are strongly incentivized to follow the protocol. We then outline a reputation management protocol which maintains this incentivisation, while reducing the probability that incorrectly performing computers are selected for computation.

\subsection{Reward Mechanism}

In this section, we provide a theoretical lower-bound on the per-computer reward required to strongly incentivize the correct participation of rational computers in our DC protocol. We describe this reward generically in terms of a cost to compute a particular computation,  fees required to post a transaction to the blockchain, and the probability of being rewarded for good and bad computations. We also describe an escrow amount to be deposited by the requester when submitting the computation, which is received back after correctly computing the set of computers to reward. Running computations such as sorting arrays or encryption/decryption on-chain is expensive, so we initially give the requester the opportunity to correctly submit the set of computers to reward. Computers have an opportunity to contest this by depositing a bounty on-chain, triggering the on-chain verification of the reward set. Note that verification is much cheaper than computation, but with respect to Privacy $\protocolName$, this reveals some private information, which is why verification does not take place automatically. If the contest is valid, the requester loses the escrow. Otherwise, the computer loses the bounty. This is described in more detail in Section \ref{sec:MarvelAlgos}.

For a given computation $\computation$, we assume an accurate a-priori lower-bound on the cost to compute a particular computation $\computation$ of $\computationCost$. This lower-bound is known by all players in the system (in reality this value can be enforced by the protocol smart-contract). Given that the payment of $\txFee$ per transaction guarantees timely inclusion in the blockchain, rational computers perform the calculation of $\computation$ if and only if:
\begin{equation} \label{eq:rewardsGreaterThanCosts}
    E(\computationReward) > \computationCost + 2\txFee.
\end{equation}

To more accurately describe $\computationReward$, we introduce $0 \leq \probGoodRating,\probBadRating \leq 1$ with the probability of good computations being rewarded being $\probGoodRating$, and the probability of bad computations being rewarded being $\probBadRating$, such that WLOG $\probGoodRating>\probBadRating$. This gives an expected payoff of $\probGoodRating \computationReward - (\computationCost +  2\txFee)>0$ for following the protocol, and an expected payoff of $\probBadRating \computationReward -  2\txFee$ for submitting a result but not performing the computation. Therefore, we require:
\begin{equation}
    \probGoodRating \computationReward - \computationCost  -\probBadRating \computationReward  >0
\end{equation}

This reduces to
$\computationReward  > \frac{\computationCost}{\probGoodRating -\probBadRating}.$
The exact values of $\probGoodRating$ and $\probBadRating$ depend on the computation, number of computers to be rewarded and the chosen target function. Exact values for  $\probGoodRating$ and $\probBadRating$ are difficult to predict a-priori. For deterministic computations $\probGoodRating \approx 1$, whereas for non-deterministic computations such as FL, $\probGoodRating$ will be smaller. Lower-bounding the possible value of $\probGoodRating -\probBadRating$ (although greater than 0) and using this value to compute $\computationReward$ ensures rational players follow the protocol.

Additionally, in a smart-contract blockchain, actions must be triggered by one or more players by sending transactions to the blockchain. In $\protocolName$, the distribution of rewards is initially controlled by the requester (if the requester fails to trigger correct reward distribution, computers can eventually initiate the rewarding mechanism). To ensure a rational requester correctly submits the set of good computations to be rewarded, any positive escrow amount $\escrowRequester>\txFee$ suffices to strongly incentivize the requester. This can be seen as the payoff for submitting the correct set is $\escrowRequester-\txFee>0$, while the payoff for not submitting the set is 0. However, in the case of potential collusion of up to $k$ computers, setting $\escrowRequester \geq k\cdot \computationReward+\txFee$ guarantees that the requester correctly submits the set of good computations. If $k$ is set too small by the protocol/smart contract, not submitting the reward set, and rewarding all players may be positive expectancy for the requester. Setting $k$ equal to $\numComputers$ conservatively achieves this.
With this lower bound on $\escrowRequester$, rational requesters always submit the correct set of good computations to the blockchain. 

\subsection{Reputation Management Protocol}\label{sec:Reputation}

In the previous section, we identified that all rational players in the system follow the protocol given no changes to reputation. However, the use of a reputation-based selection process prioritizes good computers over bad computers, meaning both short- and long-term benefits for correctly behaving computers. Therefore, using reputation-based computer selection is desirable.
In this section we describe a reputation management protocol that maintains the incentive compatibility of a DC protocol with an incentive compatible reward mechanism.

We consider computations for which a rating mechanism rates results as either ``good" or ``bad" exists. Correctly performed computations are rated good with probability $\probGoodRating$, while incorrectly computations are rated good with probability $\probBadRating$. We construct a function from this rating mechanism, $\rate()$, which assigns good calculations a score of 1, and bad calculations a score of 0. For a player $\playeri$ taking part in computations for $\computation_{1}, \computation_2,..., \computation_k $, $\playeri$'s \textit{base reputation} is $\baseRep_i=\sum_{j=1}^k \rate(\computation_j)$.

Let $\initRep>0$ be the starting reputation for computers registering in the system. For a given computation, $\playeri$ is selected as computer for a computation in block at height $\blockchainHeight$ in the blockchain in direct proportion to $\baseRep^{\blockchainHeight-1}_i (\initRep -1)$ as a fraction of $\sum^{\numPlayers}_{j=1} \baseRep^{\blockchainHeight-1}_j - \numPlayers \cdot (\initRep -1)$. We subtract ($\initRep$-1) from $\baseRep^{\blockchainHeight-1}_i$ to normalize the base reputations for selection probability purposes, and so there is no benefit for computers rejoining, particularly in the case where base reputations are increasing over time. With this in mind, the number of computations a player $\playeri$ is selected for is directly proportional to:

\begin{equation}\label{eq:SelectionProb}
    \selectionProb^\blockchainHeight_i=\frac{\baseRep^{\blockchainHeight-1}_i- (\initRep -1)}
    {\sum^{\numPlayers}_{j=1} \baseRep^{\blockchainHeight-1}_j - \numPlayers\cdot (\initRep -1)}.
\end{equation}

Consider a player $\player_i$ who includes a good computation as block proposer for a computer. This increases that computers base reputation, and thus that computers $\selectionProb$. This has long-term stake implications, as discussed at the beginning of Section \ref{sec:Reputation}. Although the computation result is encrypted, the expected change in $\selectionProb$ for an included computer is positive, which negatively affects the $\selectionProb$ of the block proposer. Therefore, we need to reward the proposer with an increase in base reputation to counteract the increase in the computers expected increase in base reputation. Let $\expectedBaseRepChange>0$ be the expected change in base reputation for a computer whose computation gets included on the blockchain. 

For $\playeri$ a proposer of a block that includes $k$ transactions containing computation results, we need the following equality to hold:
\vspace{-0.5cm}
\ignore{
\begin{equation}
\begin{split}
         \selectionProb^\blockchainHeight_i &= \frac{\baseRep^\blockchainHeight_i-(\initRep -1)}{\baseRep^\blockchainHeight_i+ \sum_{j \neq i} \baseRep^{\blockchainHeight-1}_j + k\cdot\expectedBaseRepChange - \numPlayers\cdot (\initRep -1)}  \\ 
        &= \frac{\baseRep^{\blockchainHeight-1}_i-(\initRep -1)}{\baseRep^{\blockchainHeight-1}_i+\sum_{j \neq i} \baseRep^{\blockchainHeight-1}_j - \numPlayers \cdot (\initRep -1)} .
\end{split}
\end{equation}

Solving for $\baseRep^\blockchainHeight_i$ in this equality gives:}

\begin{multline} \label{eq:reputationChange}
    \baseRep^\blockchainHeight_i= \baseRep^{\blockchainHeight-1}_i + \\  k\cdot\expectedBaseRepChange \big(  \frac{\baseRep^{\blockchainHeight-1}_i-(\initRep -1)}{\sum_{j \neq i} \baseRep^{\blockchainHeight-1}_j - (\numPlayers-1) \cdot (\initRep -1)} \big).
\end{multline}

This means we need to add $k\cdot\expectedBaseRepChange \big(  \frac{\baseRep^{\blockchainHeight-1}_i-(\initRep -1)}{\sum_{j \neq 1} \baseRep^{\blockchainHeight-1}_j - (\numPlayers-1) \cdot (\initRep -1)} \big)$ to $\baseRep^{\blockchainHeight-1}_i$ in order to ensure the proposer is impartial, with respect to reputation and computer selection probability, to adding  transactions containing computation results to the blockchain. 

For transactions from the requester finalising the rewards, we simply have to replace $\expectedBaseRepChange$ with the actual mean change in reputation in Equation $\ref{eq:reputationChange}$, and the rest of the numbers stay the same. 

\section{\texorpdfstring{$\protocolName$}{TEXT}}\label{sec:MarvelAlgos}

The goal of this section is to take the ideal DC functionality of Section \ref{sec:SINCE}, and implement it as a set of algorithms encoded as smart-contracts that can be run by a decentralized (without a TTP) set of players with access to a blockchain. We call this protocol $\protocolName$. In Section \ref{sec:SINCE}, we identified values for computer rewards, reputations changes and computer/requester escrows which ensure the participation of rational players if they can be enforced. In this section, we describe how all of these values can be enforced using a blockchain, and as such, that it is possible to implement the idealized DC protocol in a fully decentralized manner. 

\subsection{Algorithmic Overview}\label{sec:AlgoOverview}

We outline the $\protocolName$ protocol as a set of pseudo-code smart contracts encodings provided in Algorithms \ref{alg:SMRProgram} and \ref{alg:ComputerSelectionProtocol}. These contracts are labelled: \textit{Register}, \textit{Request}, \textit{Response} and \textit{Finalize}. A Solidity implementation of $\protocolName$ has also been made publicly available on Github \cite{MarvelDCGithub}.

A $\protocolName$ instance can spawn an indefinite number of computation instances, each initialized by calling the Request contract, and lasting at least $\revealTXDelay$ blocks, where $\revealTXDelay$ is the number of blocks required for players to observe an event on-chain, send a transaction and have that transaction committed on-chain given at least $\txFee$ is paid. We provide here the intuition to these encodings, including a graphic representation of the protocol flow in Figure \ref{fig:protocolFlow}. A privacy-enhancing implementation of $\protocolName$ is then described in Section \ref{sec:PrivacyMarvel}.

\begin{figure}
	\includegraphics[width=1\textwidth]{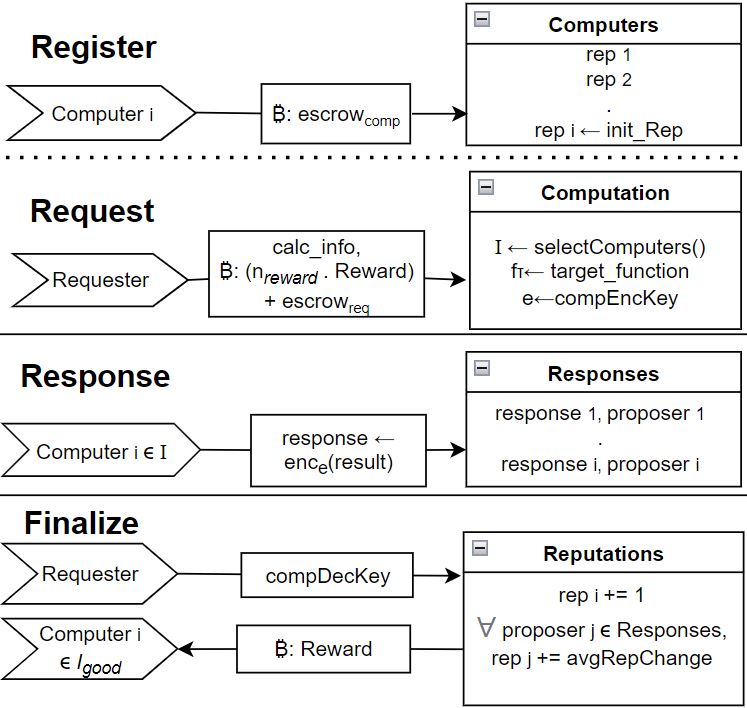}
	\caption{$\protocolName$ information and token flow. ${\bitcoin:}$ indicates the transfer of tokens.}
	\label{fig:protocolFlow}
\end{figure}

Each player $\playeri$ owns (has exclusive access to) a set of token balances $\balance_i$ which are stored as a globally accessible variable on the blockchain.
For a token $\bitcoin$, $\balance_i(\bitcoin)$ is the amount of token $\bitcoin$ that $\playeri$ owns. 
Players in the underlying blockchain protocol can enter $\protocolName$ as computers by calling the Register contract, which for a given computer deposits an escrow $\escrowComputer$ (line \ref{line:depositEscrowComp}), granting that computer a reputation of $\initRep$ (line \ref{line:initialRep}).

Computation request instances are initialized by calling the Request contract, which specifies the computation details $\computation$, the number of computers to be selected for the computation $\numComputers$, a deterministic function $\targetFunction$ for selecting the target result $\target$ from the set of results, the number of computers to reward $\numToReward$, and the per-computer reward $\playerReward$ received by a computer if included in the set of computers to reward. The requester deposits $\numToReward\cdot\playerReward +\escrowRequester$. The $\compEncKey$ is the public key corresponding to the temporary public/private key pair ($\compEncKey, \compDecKey$). This is a key pair generated by the requester specifically for $\computation$. A randomness beacon is called (line \ref{line:getComputers}), which provides a pseudo-random seed for selecting $\numComputers$ computers to participate in the computation in direct proportion to computer reputations (line \ref{line:func:SelectComps}), with these computers listed in the set $\computation.\indicesForComputation$. 

Computers selected for a particular computation $\computation$, identified in $\computation.\indicesForComputation$, can then submit results for $\computation$ to the blockchain by calling the Response contract for up to $\revealTXDelay$ blocks after the computation is requested. These results should be encrypted using $\computation.\compEncKey$ (line \ref{line:SubmitResponse}). These encrypted results are decrypted in the Finalize phase (line \ref{line:decrypt}) for rewarding purposes. This encryption ensures no other computer can use another computer's result, and therefore must themselves perform the computation. Given a valid response is recorded, the block producer corresponding to the response is added to $\computation.\computationProposers$ (line \ref{line:addTXProposer}). This is used later to update reputations, in line with the analysis of Section \ref{sec:Reputation}.

Then, either after $\revealTXDelay$ blocks from when the computation $\computation$ was requested, or when all computers in $\computation.\indicesForComputation$ have responded, the requester of $\computation$ can complete the request by calling the Finalize contract. Calling the Finalize contract requires the requester to provide the decryption key $\computation.\compDecKey$ corresponding to $\computation.\compEncKey$. If these keys match up, the requester receives back her escrow $\escrowRequester$. The contract then uses $\compDecKey$ to decrypt the computer responses, and identify which computers are to be rewarded (line \ref{line:func:rateComputations}). This is done by applying the pre-specified target function to the computation results, and computing a target value $\target$ (line \ref{line:setTargetResult}). The computers corresponding to the $\computation.\numToReward$ results closest to $\target$ using Euclidean distance are selected as the computers to reward, $\goodComputers$ (line \ref{line:getBestResults}). The computers in $\computation.\goodComputers$ each receive $\computation.\playerReward$. Finally, all registered computers in $\computation.\computationProposers$ receive reputation increases of $\avgRepChange$ (line \ref{line:func:updateReps}), while computers in $\computation.\goodComputers$ each receive an increase in reputation of 1 (line \ref{line:goodCompRepChange}).

\begin{algorithm}
\def\baselinestretch{1} \scriptsize \raggedright
\caption{\texorpdfstring{$\protocolName$}{TEXT} smart contract pseudocode}
\label{alg:SMRProgram} 
	\begin{algorithmic}[1] 
	    
	    \State $\computers \assign [] $ \Comment{Set of active computers}
	    \State $\initRep \assign \textit{getInitialRep}()$ 
	    \State $\reputations \assign [\initRep \ \textbf{for} \ i \in \computers]$
	    
	    \State $\revealTXDelay \assign \getFinalizeDeadline()$ \Comment{Globally-defined finalize deadline} \label{line:finalizeDeadline}
	    \State $\escrowComputer, \escrowRequester \assign \getEscrows()$ \Comment{Globally-defined escrow amounts, in line with Section \ref{sec:SINCE}}
	    \State $\numComputersBound \assign \textit{getMinNumComputersPerComputation}()$ \Comment{Set $\numComputersBound$ in-line with requirements from Section \ref{sec:SINCE}}
	    \State $\targetFunctions \assign \textit{getTargetFunctions}()$ \Comment{Define allowable target functions. In reality, this can be updated during the protocol}

	    \SPACE 
	    
	    \State \textbf{Register}
	    
	    \Upon{$\langle \tagRegister \rangle \ \From  \ \player \ \With \ \player \notin \computers \  \logicalAnd \ \player\getBalance > \escrowComputer $}{}\label{func:Register}\Comment{add computer to the system}
	        \State $\player\getTransfer(\escrowComputer , \contract)$ \Comment{Registration cost to prevent Sybil attacks} \label{line:depositEscrowComp}
		    \State $\computers\getAppend(\player)$
		    \State $\reputations\getAppend(\initRep)$  \label{line:initialRep}
		\EndUpon
	    
	    \SPACE 
	    
	    \State \textbf{Request}
	    
		\Upon{$\langle \tagRequest , \computation, \numComputers, \numToReward, \compEncKey, \playerReward, \targetFunction \rangle \ \From  \ \requester \ \With \ \numComputers> \numComputersBound \ \logicalAnd \  \requester\getBalance >  ((\numToReward,\cdot \playerReward) + \escrowRequester ) $ $\logicalAnd$ $ \targetFunction \in \targetFunctions $}{}\label{func:Request}
		    \State $\requester\getTransfer((\numToReward, \cdot \playerReward) + \escrowRequester, \contract)$ \Comment{Transfer total reward plus requester escrow to contract}
		    \State $\indicesForComputation \assign \selectComputers(\getRandomness(), \numComputers)$ \Comment{Select computers for computation} \label{line:getComputers} 
		   
		    \State $\responses \assign []$
		    \State $\computationProposers \assign []$ \Comment{Array of the players who recorded each $\langle \tagResponse, * \rangle $ transaction}
		    \State $\computationStart \assign \currentHeight$  \Comment{Record current height of blockchain}
		    \State $\step  \assign \stepComputing$
		\EndUpon

		\SPACE 
		
		\State \textbf{Response}
		
		\Upon{$\tx \assign \langle \tagResponse, \computation, \computationResult \rangle \ \From \ \computer \in \indicesForComputation $ $\ \With \  \computation.\step = \ \stepComputing \ \logicalAnd \ \currentHeight $ $ \ < \computation.\computationStart + \revealTXDelay  \  $  } \label{func:RESPONSE}
			\State $\computation.\responses\getAppend(\computationResult) $ \Comment{$\computationResult$ should be the computer $\computer$'s result of computing $\computation$, encrypted using $\computation.\compEncKey$}
			\If{$\tx.\blockProposer \in \computers$} 
			    \State $\computation.\computationProposers\getAppend(\tx.\blockProposer)$ \label{line:addTXProposer}
			\EndIf
		\EndUpon
		
		\SPACE 
		
		\State \textbf{Finalize}
		
		\Upon{$\tx \assign \langle \tagFinalize, \computation, \compDecKey \rangle \ \From \ \computation.\requester $ $\ \With \ \isValid(\compDecKey, \computation.\compEncKey) $ $ \ \logicalAnd \  \big( ( \computation.\step = \ \stepComputing \ \logicalAnd \ \currentHeight $ $ \ < \computation.\computationStart + \revealTXDelay  \ $ $\logicalAnd \ \getLength(\computation.\responses)=\computation.\numComputers ) $ $ \ \logicalOr \ (\computation.\step = \ \stepComputing \ $ $\logicalAnd \ \currentHeight  \ \geq  \computation.\computationStart + \revealTXDelay) \big) $  } \label{func:FINALIZEsuccess}
		    \State $\computation\getTransfer(\escrow, \computation.\requester)$ \Comment{Returns the escrow to the requester} \label{line:requesterRefund}
		    \State $\computation.\computationProposers\getAppend(\computation.\requester)$ \Comment{Required for SINCE of requester}
		    \If{$\tx.\blockProposer \in \computers$} \Comment{Required for SINCE of proposers}
			    \State $\computation.\computationProposers\getAppend(\tx.\blockProposer)$
			\EndIf
			\State $\computation.\step  \assign \stepFinalized$
			
			\State $\goodComputers, \badComputers \assign \rateComputations(\computation, \computation.\numToReward, \compDecKey, \computation.\targetFunction)$ \Comment{Function which deterministically evaluates the goodness of returned computations, returning the indices of good and bad computers}
			\State $\computation\getTransfer(\computation.\playerReward, \goodComputers)$ \label{line:payGoodComps}
			\State $\updateReputations(\goodComputers, \badComputers, \computation) $  \label{line:updateReps}
			
		\EndUpon
		
		\SPACE 

	\end{algorithmic} 
	
\end{algorithm}

\begin{algorithm}
\def\baselinestretch{1} \scriptsize \raggedright
\caption{Computer Selection Protocol}
\label{alg:ComputerSelectionProtocol} 
	\begin{algorithmic}[1]

        \Function{$\genProbSelect()$}{} \label{line:func:generateSelctionProbabilities}
		\State $\minRep \assign \getMinimum(\reputations)$
		\State $\denominator \assign \sum(\reputations) - \getLength(\reputations)\cdot(\minRep-1)$
		\State \textbf{return} $[\big( (i- (\minRep-1))/ \denominator \big) \ \textit{for} \ i \in \reputations]$ \Comment{Probability formula from Section \ref{sec:Reputation}}
		\EndFunction
		
		\SPACE
		
		\Function{$\selectComputers(\randomSeed,\numComputers)$}{} \label{line:func:SelectComps}
		\State $\counter \assign 0$
		\State $\indicesForComputation \assign []$
		\State $\probSelect \assign \genProbSelect(\reputations)$
		\State $\randomSeed \assign \hash(\randomSeed)$
		\While{$\counter < \numComputers$} \label{line:drawfromUrn}
		    \State $i\assign 0$
		    \State $\sumReps \assign \probSelect[i]$
		    \While{$\randomSeed > \sumReps$} \label{line:findComputer}
		        \State $\sumReps\assign \sumReps + (\probSelect[i++]*(2^256))$
		    \EndWhile
		    \If{$\logicalNot(i \in \indicesForComputation)$} \label{line:uniqueCompFound}
		        \State $\indicesForComputation\getAppend(i)$
		        \State $\counter \assign \counter +1 $
		    \EndIf
		    \State $\randomSeed \assign \hash(\randomSeed)$
		\EndWhile
		\EndFunction

	\end{algorithmic} 
	
\end{algorithm}

\begin{algorithm}
\def\baselinestretch{1} \scriptsize \raggedright
\caption{Reputation Management}
\label{alg:RepMgmtProtocol} 
	\begin{algorithmic}[1]

		\Function{$\rateComputations(\computation, \numToReward,  \compDecKey, \targetFunction)$}{} \label{line:func:rateComputations}
		\State $\goodComputers \assign []$
		\State $\results \assign \decrypt(\computation.\responses, \compDecKey)$ \label{line:decrypt}
		\State $\target \assign \targetFunction(\results)$ \label{line:setTargetResult}
		\For{$ i \ \in \ [1,..., \numToReward]$} \Comment{add the $\numToReward$ closest results to $\target$ to $\goodComputers$} \label{line:getBestResults}
		  \State $\goodComputers\getAppend(\result.\computer) \logicalAnd \results.\textit{remove}(\result)$ $\With \getDistance(\result,\target) = $ $ \getMinimum(\getDistance(\results,\target))$
		  
		\EndFor
		\State $\badComputers \assign \results.\computer$ \Comment{all results not already removed in the for loop are bad results, not to be rewarded}
		\State \textbf{return} $\goodComputers, \ \badComputers$
		\EndFunction
		
		\SPACE
		
		\Function{$\updateReputations(\goodComputers, \badComputers, \computation)$}{}
		\State $\avgRepChange \assign \getLength(\goodComputers)/(\getLength(\goodComputers)+\getLength(\badComputers))$
		\State $\denominator \assign \sum(\reputations) - (\getLength(\reputations)-1)\cdot(\initRep-1)$
		\For{$\blockProposer \ \in \ \computation.\computationProposers$} \label{line:func:updateReps} \Comment{in-line with the results from Section \ref{sec:Reputation}, block proposers rep. changes should be done before updating computers}
		    \State $\reputations[\blockProposer] \assign $ $\reputations[\blockProposer]+ $ $ (\avgRepChange \cdot \big( (\reputations[\blockProposer]-$ $ (\initRep-1))/ (\denominator-\reputations[\blockProposer]) \big))$ \Comment{Necessary for SINCE of proposers/requester}
		\EndFor
		
		\State $\reputations[\computation.\requester] \assign $ $\reputations[\computation.\requester] + $ $ (\avgRepChange \cdot \big( (\reputations[\blockProposer]-$ $ (\initRep-1))/ (\denominator-\reputations[\blockProposer]) \big))$ \Comment{Requester of successfully resolved computation must receive increase in reputation, in line with Section \ref{sec:Reputation}} 
		\State $\reputations[\tx.\blockProposer] \assign $ $\reputations[\tx.\blockProposer] + $ $ (\avgRepChange \cdot \big( (\reputations[\blockProposer]-$ $ (\initRep-1))/ (\denominator-\reputations[\blockProposer]) \big))$ \Comment{Proposer including the Finalize transaction must also receive increase in reputation, in line with Section \ref{sec:Reputation}}

		\State $\reputations[\goodComputers] \assign \reputations[\goodComputers]+1$ \label{line:goodCompRepChange}

		\EndFunction
		
	\end{algorithmic} 
	
\end{algorithm}

\begin{algorithm}
\def\baselinestretch{1} \scriptsize \raggedright
\caption{\texorpdfstring{$\protocolName$}{TEXT} for player $\playeri$ as a Requester}
\label{alg:RequesterProtocol} 
	\begin{algorithmic}[1] 
	    
	    \Function{$initialize(\computation)$}{}\label{func:Initialise}
	        \State $\compEncKey, \compDecKey \assign \textit{generateKeyPair}()$
	        \State $\playerReward \assign $ SINCE reward to guarantee participation of computers 
	        \State $\numComputers \assign x \ \With x >  \numComputersBound$ 
	        \State $\targetFunction \textit{getTargetFunction}(\computation)$ \Comment{Select target function for computation}
	        \State \Broadcast $\langle \tagRequest , \computation, \numComputers, \numToReward, \compEncKey, \playerReward, \targetFunction \rangle$
		\EndFunction

		\SPACE 
		\Upon {$ \getLength(\computation.\responses)=\computation.\numComputers \ \logicalOr \ ( \currentHeight  = \computation.\computationStart + \finalizeDeadline ) $  } \label{func:RequesterFINALIZE}
			\State \Broadcast $\langle \tagFinalize, \computation, \computation.\compDecKey \rangle$
		\EndUpon
		
	\end{algorithmic} 
	
\end{algorithm}

\begin{algorithm}
\def\baselinestretch{1} \scriptsize \raggedright
\caption{\texorpdfstring{$\protocolName$}{TEXT} for player $\playeri$ as a computer}
\label{alg:ComputerProtocol} 
	\begin{algorithmic}[1] 
	    
	    \State \Broadcast $ \langle \tagRegister \rangle $
	    
		\Upon{$\langle \tagRequest , \computation, \numComputers, \numToReward, \compEncKey, \playerReward, \targetFunction \rangle \ \With \ i \in \computation.\indicesForComputation$ }{}\label{func:ComputerAction}
		    
		    \State \Broadcast $ \langle \tagResponse, \computation, \computationResult \assign \encrypt(\compute(\computation), \computation.\compEncKey) \rangle $ \label{line:SubmitResponse}
		\EndUpon
		
	\end{algorithmic} 
	
\end{algorithm}

\subsection{Protocol Properties}\label{sec:properties}

In the following results, we make the assumption that rational requesters randomly enter the system, running unique instances of the Request contract. Under this assumption, we first show that rational computers and rational requesters are strongly incentivized to participate in the protocol. 

\begin{spacing}{1.5}
\begin{theorem}\label{thm:SINCE}
    There is a strict Nash Equilibrium in which, for any computation with a per player reward $\playerReward > \frac{\computationCost}{\probGoodRating -\probBadRating}$, rational computers and requesters follow the protocol.
\end{theorem}

\begin{proof}
    Consider a Request( $\requester, *$) instance corresponding to a computation $\computation$, and computers selected for computation $\indicesForComputation$. Based on $\numComputers> \numComputersBound$, the majority of computers in $\indicesForComputation$ are rational. 
    
    First consider a rational requester. Correctly running Finalize ($ \computation, *$) allows the requester to receive back $\computation.\escrowRequester$, and as such, rational requesters follow the protocol.
    
    Consider now rational computers. If the requester correctly runs Finalize ($ \computation, *$), then $\computation.\target$ and $\computation.\goodComputers$ are generated correctly. Therefore, if all rational computers follow the protocol, the assumption under which we chose $\playerReward$ in Section \ref{sec:SINCE}, for a given rational computer $\computer_i$ correctly running Repsonse($ \computation, $ $ *$), $\computer_i$ is included in $\computation.\goodComputers$ with probability $\probGoodRating$. If $\computer_i$ incorrectly runs Response($ \computation, $ $ *$), $\computer_i$ is included in $\computation.\goodComputers$ with probability of at most $\probBadRating$. By our choice of $\playerReward$, we have seen in Section \ref{sec:SINCE}, given $\computation.\goodComputers$ is generated correctly and computers included in $\computation.\goodComputers$ receive this with probability 1, this is sufficient for rational computers to compute the result correctly, equivalent to calling Response($ \computation, *$).
    
    Therefore, rational computers and requesters follow the protocol if $\playerReward > \frac{\computationCost}{\probGoodRating -\probBadRating}$
\end{proof}
\end{spacing}

This result is enough to ensure rational players follow the $\protocolName$ protocol. However, because of the use of the same reputation and computer-selection mechanism as described in Section \ref{sec:Reputation}, $\protocolName$ also guarantees that Byzantine computers are selected with diminishing probability in the number of computations, converging to 0 for any minority of selected computers. This is stated formally in the following lemma. 

\begin{spacing}{1.5}
\begin{lemma}\label{lem:diminishingProb}
    For a series of computations $[\computation_1, \computation_2, ..., \computation_i]$ with $\playerReward > \frac{\computationCost}{\probGoodRating -\probBadRating}$ and $\numComputers >  \numComputersBound$, as the number of completed computations increases, the probability of selecting a Byzantine computer for a computation with $\numComputers < \frac{|\computers|}{2}$ is strictly decreasing in expectancy and converging to 0 as $i$ tends to infinity. 
\end{lemma}

\begin{proof}
    As  $\playerReward > \frac{\computationCost}{\probGoodRating -\probBadRating}$, from Theorem \ref{thm:SINCE} rational computers follow the protocol. Let $\adversarialBound$ be the share of computers that are Byzantine. We know a majority of computers selected are rational, as $\numComputers >  \numComputersBound$. Therefore, Byzantine computers are rewarded with probability $\probBadRating < \probGoodRating$. For a given computation, the expected reputation increase of a selected Byzantine computer is $\probBadRating$, while the expected increase for a selected rational computer is $\probGoodRating$. 
    Given $\numComputers$ are selected for the computation, the expected number of these being rational computers is  $(1-\adversarialBound) \numComputers$, while the number of selected Byzantine computers is $\adversarialBound \numComputers$. Furthermore, this means the expected increase in reputation for rational computers is $(1-\adversarialBound) \numComputers \probGoodRating $, while the expected increase in reputation for Byzantine computers is $\adversarialBound \numComputers\probBadRating $. At the beginning of the protocol, the probability of selecting a Byzantine player from the set of all computers is in direct proportion to starting reputation. Given initial reputations of $\initRep$, after the first computation, the selection probability of a Byzantine computer reduces in expectancy to:
    \begin{equation}
        \frac{\adversarialBound (|\computers| \cdot \initRep + \numComputers\probBadRating)}{|\computers| \cdot \initRep + \numComputers \big( (1-\adversarialBound)\probGoodRating +\adversarialBound\probBadRating \big)}.
    \end{equation}
    First it be can see that 
    \begin{equation}
        \frac{\adversarialBound (|\computers| \cdot \initRep + \numComputers\probBadRating)}{|\computers|\cdot \initRep + \numComputers \big( (1-\adversarialBound)\probGoodRating +\adversarialBound\probBadRating \big)}<\adversarialBound
    \end{equation}
    meaning Byzantine selection probability is decreasing. To prove that Byzantine selection probability tends to 0 in the number of computations as described in the Lemma statements, let $\adversarialBound_k$ be the Byzantine computer selection probability after $k$ computations. We have the expected Byzantine selection probability after $k+1$ computations, denoted , $\adversarialBound_{k+1}$, is:
    \begin{equation}
    \begin{split}
        \frac{\adversarialBound_k (|\computers| \cdot \initRep + \numComputers\probBadRating)}{|\computers| \cdot \initRep + \numComputers \big( (1-\adversarialBound_k)\probGoodRating +\adversarialBound_k\probBadRating \big)} \\
        = \frac{\adversarialBound_k (|\computers| \cdot \initRep + \probBadRating\numComputers)}{|\computers| \cdot \initRep + \numComputers \probGoodRating - \adversarialBound_k \numComputers (\probGoodRating-\probBadRating )}.
    \end{split}
    \end{equation}
    We have already seen $\adversarialBound_{k+1}$ equals
    \begin{equation}
        \frac{\adversarialBound_k (|\computers| \cdot  \initRep + \numComputers\probBadRating)}{|\computers| \cdot \initRep + \numComputers \probGoodRating - \adversarialBound_k \numComputers (\probGoodRating-\probBadRating )}< \adversarialBound_k.
    \end{equation}
    which implies:
    \begin{equation}
        \frac{ (|\computers| \cdot \initRep + \numComputers\probBadRating)}{|\computers| \cdot \initRep + \numComputers \probGoodRating - \adversarialBound_k \numComputers (\probGoodRating-\probBadRating )}<1.
    \end{equation}
    Letting the term on the left be $r_k$, we can see $r_k$ is decreasing in $k$ as:
    \begin{itemize}
        \item $\numComputers (\probGoodRating-\probBadRating )>0$ (because $\probGoodRating>\probBadRating$).
        \item $0<\adversarialBound_{k+1}<\adversarialBound_k$.
    \end{itemize} 
    These together mean the negative term in the denominator of $r_k$, $\adversarialBound_k \numComputers (\probGoodRating-\probBadRating )$, is increasing (towards 0) and as such the denominator of $r_k$ is increasing. Therefore $\adversarialBound_k< \adversarialBound_0 r_{0}^k$, with $r_0 <1$. The result follows. 
\end{proof}

\begin{remark}
    Lemma \ref{lem:diminishingProb} depends on the output of the on-chain randomness oracle being unpredictable when Request is called. Existing solutions, such as the Chainlink VRF \footnote{\url{https://docs.chain.link/docs/chainlink-vrf/}}, provide proofs that provided randomness was generated correctly. Analysis of the quality of this randomness is beyond the scope of this work. 
\end{remark}
\end{spacing}

As a direct consequence of Lemma \ref{lem:diminishingProb}, with reasonable choices for rewarding functions and number of computers per-computation (explored in Table \ref{table:advShares}), both enforceable by the protocol, Byzantine players are eventually removed from the system. This improves the efficiency of the protocol over time, reducing the minimum requirements for computers, and as such, latency, transaction fees, and rewards.

\subsection{Privacy \texorpdfstring{$\protocolName$}{TEXT}}\label{sec:PrivacyMarvel}

In this section we outline a privacy enhancement to $\protocolName$ which we call Privacy $\protocolName$. The motivation for this enhancement is to allow for an additional level of computer privacy which can be seen as necessary in computations such as those in FL protocols. The privacy provided is based on existing, well-known ZK techniques. However, this additional privacy on top of the novel contributions of being strong incentive compatible, generically applicable and fully decentralized further add to the applicability and utility of our work in an even larger set of DC problems.

We present Privacy $\protocolName$ by describing it's key differences to $\protocolName$ to ensure that in an optimistic scenario, only the requester and computers involved in a computation learn the results, and that players in the system can at most infer a computer submitted a good result (or bad result), but not which of the good results (bad results). In the pessimistic scenario, all players in the blockchain observe the results, but it still holds that any player in the system can at most infer a computer submitted a good result (or bad result), but not which of the good results (bad results). In Privacy \texorpdfstring{$\protocolName$}{TEXT}, there is an additional contract, \textit{Reveal}, which is to be executed after the Response contract, and before rewards are finalized. The purpose of the Reveal contract is described later in this section.

During computer registration, computers in Privacy $\protocolName$ privately generate $\commSerialNum_1, \ $ $ \commRandomness_1 \in \{0,1\}^{O(\cryptoParam)}$, and attach $\regToken_1 \assign $ $ \commit(\commSerialNum_1,\commRandomness_1)$ to the registration message, as described in Section \ref{sec:ZK}. Then, when a requester requests a computation, and the indices for computation $\indicesForComputation$ are calculated, the requester now generates a Merkle Tree containing the indices as specified in $\indicesForComputation$. This Merkle Tree is the structure to which only selected computers submitting results can prove membership in ZK. Results are therefore associated to a ZK set-membership proof which reveals nothing about the player that generated the proof. In Privacy $\protocolName$, this separates result submission and player identity. To maintain this separation of identity from result, ZK set-membership proofs are required in the updated Finalize contract, described later in this section. 

In addition to the deposits of $\protocolName$, the requester must also deposit a pool of money necessary to incentivize relayers (as described in Section \ref{sec:Relayers}) to publish transactions on behalf of computers involved in the computation. Given the amount of money required by one relayer to include a blockchain transaction is $\feeRelayer$, the additional required deposit is $\numComputers \cdot \feeRelayer$ for the relaying of computer messages during the Respnse phase.

In the Response contract for Privacy $\protocolName$, computers selected in $\indicesForComputation$ privately generate a ZK-proof proving their membership in $\indicesForComputation$. Formally, given a computation result of $\result$, the computer sets $\response \assign \encrypt $ $(\compute(\computation)$ $,\computation.\compEncKey)$. The computer also generates a new $\commSerialNum_2, \ $ $ \commRandomness_2 \in \{0,1\}^{O(\cryptoParam)}$ pair, and computes $\regToken_2 \assign $ $ \commit(\commSerialNum_2,\commRandomness_2)$ . Setting $m \assign \langle \computation, \response, \commRandomness_1, \regToken_2 \rangle$, the computer generates a NIZKSoK  $\proofZK_1 \assign $ \textit{NIZKSoK}[$m$]$\{(\regToken_1,$ $\commRandomness_1):$ MemVerify ($\indicesForComputation,$ $ \regToken_1 )$ $=1 $ $ \And $ $ \regToken_1=$ $ \commit($ $\commSerialNum_1, $ $\commRandomness_1) $ $\}$. 
Finally, the computer then publishes $m$ and $\proofZK_1$ to the blockchain through a relayer, who receives $\feeRelayer$ upon the transactions addition to the blockchain.  

In the Reveal contract, the requester off-chain performs the same calculations that were done on-chain in $\protocolName$ to calculate the results to be rewarded, but instead of adding computer indices to $\goodResponses$, the requester adds the corresponding $\regToken_2$s. The requester publishes $\goodResponses$ and the encryption of $\computation.\compDecKey$ using each public key corresponding to computers in $\computation.\indicesForComputation$ to the blockchain. However, rewards are not immediately distributed to computers in $\goodResponses$.

In the Finalize contract, computers now have a chance to contest the computation of $\goodResponses$ for up to $\revealTXDelay$ blocks after the Reveal contract is called. If $\goodResponses$ was computed incorrectly, any of the computers in $\computation.\indicesForComputation$ can publish the decryption of all results and pass them into the $\rateComputations$ function of $\protocolName$, proving the incorrect computation of $\goodResponses$ by the requester. In this case, all computers are rewarded, and the requesters escrow is destroyed. To prevent malicious computers in $\computation.\indicesForComputation$ from attempting this proof in order to reveal computation results, a further escrow is required, which is returned on the correct proving of miscomputation of $\goodResponses$ by the requester.

If instead $\goodResponses$ was computed correctly, any computer whose $\regToken_2$ is included in $\goodResponses$ can generate a proof of membership to $\goodResponses$. Furthermore, as $\regToken_1$ can no longer be used for future computations (using the same $\regToken_1$ would reveal the same $\commRandomness_1$ in the next $\computation$), the computer generates a new $\commSerialNum_3, \ $ $ \commRandomness_3 \in \{0,1\}^{O(\cryptoParam)}$ pair and corresponding $\regToken_3 \assign $ $ \commit(\commSerialNum_3,\commRandomness_3)$. Setting $m\assign\langle \computation, false, \regToken_3 \rangle $, the computer with $\regToken_2$ in $\goodResponses$ then generates a NIZKSoK  $\proofZK_2 \assign $ \textit{NIZKSoK}[$m$]$\{(\regToken_2,$ $\commRandomness_2):$ MemVerify ($\goodResponses,$ $ \regToken_2 )$ $=1 $ $ \And $ $ \regToken_2=$ $ \commit($ $\commSerialNum_2, $ $\commRandomness_2) $ $\}$. Finally, the computer publishes $m$ and $\proofZK_2$ to the blockchain. In this variation of the protocol, computers not included in $\goodResponses$ must also respond, updating there $\regToken_1$. If this is not done, computers forfeit $\escrowComputer$ and any reputation earned. 

\begin{spacing}{1.5}
\begin{observation}

    As computers submit results through a relayer, and with an accompanying NIZKSoK $\proofZK$ proving membership in the selected indices for computation, all players in the blockchain protocol can be sure the player submitting the message must be a selected computer, but nothing else can be learned about the submitting player's identity. Similarly, when collecting rewards, or replacing $\regToken_1$, the only thing learned when a computer submits a valid message during the Finalize phase is to which set, $\goodResponses$ or $\responses \backslash \goodResponses$, the computer belongs.
\end{observation}
\end{spacing}

\subsection{Further Privacy Enhancements}\label{sec:FurtherPrivacyRequirements}

There are further privacy enhancements possible for $\protocolName$. One such enhancement is to detach reward collection/reputation updates from the computation. Given $\goodResponses$ was calculated correctly, computers included in $\goodResponses$ can instead delay their claiming of the reward and associated reputation increase arbitrarily. After a Finalize phase without arbitration, the set of $\regToken_2$s corresponding to $\goodResponses$ can be added to a pool of all recorded good responses throughout the protocol. These can then be immediately/periodically/sporadically claimed by computers, depending on the privacy requirements of the computer in question. This again uses the same NIZK set-membership techniques, except now with a larger set in which to diffuse.

\section{Implementation Analysis}\label{sec:Analysis}

In this section we analyse $\protocolName$ and Privacy $\protocolName$ as described in Sections \ref{sec:MarvelAlgos} and \ref{sec:PrivacyMarvel}. We show that on top of the unique formal guarantees of Section \ref{sec:properties}, both protocols are cost-effective and practical for both computers and requesters. Our reputation and computer selection mechanisms perform analogously to those graphed in \cite{CoUtileP2PDecentralizedComputing}. Crucially however, our protocols are proven to be strong incentive compatible in a fully decentralized setting, where computers an requesters have asymmetric utilities. The encoding of $\protocolName$ is available here \cite{MarvelDCGithub}.

\subsection{Gas cost of running \texorpdfstring{$\protocolName$}{TEXT} and Privacy \texorpdfstring{$\protocolName$}{TEXT}}\label{sec:CostAnalysis}

\begin{table}[H]
\begin{center}
    
\scalebox{0.7}{
  \begin{tabular}{ |l|l|l|}
    \toprule
     & $\protocolName$ & Privacy $\protocolName$  \\
    \bottomrule
     Setup() & 2,000,000 & 2,450,000 + $h_\secParam \cdot 150,000 ^{\dagger}$ \\
     \midrule
     Register() & 80,000 & 82,000 \\
     \midrule
     Request() & 500,000 + $\numComputers$ $\cdot$ 31,000 + 0.25 LINK$^{\dagger \dagger}$ & \makecell{500,000 + $\numComputers$ $\cdot$ 31,000 \\ + $\numComputers$ $\cdot$ $h_\secParam \cdot 51,000^{\dagger \dagger\dagger}$ + 0.25 LINK$^{\dagger \dagger}$} \\
     \midrule
     Response()& 60,000  & 63,000 +$300,000^{\dagger \dagger \dagger \dagger}$ \\ 
    \midrule 
    Reveal() & N/A & 80,000 +$\numToReward$ $\cdot$ $h_\secParam \cdot 51,000^{\dagger \dagger\dagger}$ \\
    \midrule
    Finalize() & - & - \\
    \midrule
     -No Arbitration & 80,000 +$\numToReward$ $\cdot$ 14,000 & (results not revealed) 83,000 + $300,000^{\dagger \dagger \dagger \dagger}$ \\
     \midrule
     -Arbitration & N/A & 80,000 +$\numComputers$ $\cdot$ 14,000 \\
     
    \bottomrule
  \end{tabular}}
  \caption{Gas costs of each contract call in $\protocolName$ and Privacy $\protocolName$. \footnotesize{$^\dagger$Cost of accessing Chainlink randomness. $^{\dagger\dagger}$Cost to generate a once-off Merkle Tree which can be reused for all computations. $^{\dagger\dagger\dagger}$ Cost to insert computer into copy of Merkle Tree. $^{\dagger\dagger\dagger\dagger}$ Cost to verify proof of membership on-chain.} } \label{table:gascosts}
  
\end{center}
\end{table}

There are several considerations when calculating the cost of running Marvel DC on a blockchain. In the case where a computation has a single 32 byte answer, the costs are significantly less than in the case of gradient estimation problem where answers contain thousands or millions of numbers. More concerning still is the prohibitive nature of messages in the order of MBs in many blockchain protocols. To counteract this, messages for the Response and Finalize contracts can be submitted to memory-efficient alternatives such as IPFS\footnote{\url{https://ipfs.io/}}, Layer-2 chains or even through an MPC protocol between computers and requesters (not necessarily the same entities involved in the computation). All of these suggestions require the use of checkpoints (such as those proposed in \cite{Casper}), pieces of summary information identifying permanent states in an off-chain protocol which are published to the main-chain, updating the state of the main-chain.

In Table \ref{table:gascosts} we present the approximate gas costs for $\protocolName$ and Privacy $\protocolName$ on an Ethereum-compatible blockchain of the main operations for a requester and computers in a computation, given $\numComputers$ and a 256-bit result, with all messages published on-chain. This can be extended to $l$-dimensional results for any $l>1$. The set-membership tools described in Privacy $\protocolName$ are pre-compiled, and currently being used in the Tornado Cash privacy protocol. We thus calculate the gas cost of the set-membership tools using the Tornado Cash library. All other operations are typical on-chain array manipulation, encryption/decryption, deposit, withdraw and writing to memory operations. We also include the cost of calling an on-chain randomness oracle in our calculation. This call needs to be made prior to the calculation of the indices for computation, and must be made when depositing rewards and escrow to ensure the requester cannot manipulate the selection of computers. To approximate the costs of running Privacy $\protocolName$ described in Section \ref{sec:PrivacyMarvel}, we utilize $\numComputersBound$ the maximum number of computers required to select in order to guarantee a majority are rational. We let $h_\secParam = \lceil\textit{log}_2(\numComputersBound)\rceil+1$. This is the required height of the Merkle Tree to be used in the ZK-proofs for set-membership.

As an example, using current gas costs of 32 GWEI/gas and 1 GWEI = $\$0.00000335$ \footnote{\url{https://etherscan.io/gastracker}, Accessed: 03/04/2022}, for $\numComputersBound=10$, and a computation with $\numComputers=10$, for the basic protocol, this is a per computer cost of 140,000 gas, or $\$15$, while a requesters cost is 950,000 gas, or $\$102$ + $\$4.5$ for the LINK\footnote{\url{https://docs.chain.link/docs/vrf-contracts}, Accessed:03/04/2022 } required to access the on-chain randomness \footnote{\url{https://www.coingecko.com/en/coins/chainlink}, Accessed:03/04/2022}. Deploying on a cheaper protocol like Harmony with gas costs of 
$\$0.00000005$ at the time of writing \footnote{\url{https://explorer.harmony.one/}, Accessed: 03/04/2022}, this is an approximate cost of $\$0.001$ and $\$4.505$ respectively. 

For Privacy $\protocolName$ on Harmony with $\numToReward=5$, this gives $h_\secParam=4$, and a gas cost of 950,000+ 2,040,000 for requesting, and 80,000+1,020,000 for revealing rewards, a total of 4,190,000 gas ($\$0.02$), plus the cost for accessing on-chain randomness ($\$4.5$), a total cost for the requester of $\$4.52$ . For computers, the cost is 746,000 gas ($\$0.004$) in the optimistic execution, and 503,000 ($\$0.0026$) in the pessimistic, but only for one computer. This means both implementations of $\protocolName$ are suitable for immediate deployment on public blockchains. 

\subsection{Latency and Throughput}

A direct comparison of our protocol to most existing distributed FL solutions with respect to throughput and latency is of limited benefit. This is because the the concept of time in decentralized systems (blocks in a blockchain) has no immediate parallel when centralized parties are required, as in \cite{BlockchainDecentralisedFedLearningLi,FLChainBao,FLChainMajeed,FedLearningMeetsBlockchainMa,BlockFLKim,DeepChainWeng, BlockchainRepAwareFedLearningRehman}.

Considering the protocols which allow for tokenized rewards in a decentralized setting (protocols which do not allow for tokenized rewarding have limited utility in a decentralized setting), this leaves only \cite{IncentiveMechFLBlockchainToyoda} and \cite{FairnessIntegrityPrivacyBlockchainBasedFLSystemRuckel}. Using the terminology of our paper, every protocol \textit{phase}, a period of time where an event occurs which requires a response, lasts up to $\revealTXDelay$ blocks.  These $\revealTXDelay$ blocks (as used in Section \ref{sec:AlgoOverview}) are equivalent to the time required to ensure players can submit transactions to the blockchain after a particular on-chain event, such as a computation request. $\protocolName$ therefore lasts up to $2 \revealTXDelay$ blocks, which covers the time for computers to respond to a request, and the time taken for the requester to reveal the decryption key. The protocol of \cite{IncentiveMechFLBlockchainToyoda} lasts at least $4 \revealTXDelay$ blocks to ensure workers are incentivized to submit at least one correct model update (using the terminology of \cite{IncentiveMechFLBlockchainToyoda}, 2 model update rounds are needed for this to be the case). The additional costs are due to requesters and computers being required to respond twice each after the initial request. The protocol of \cite{FairnessIntegrityPrivacyBlockchainBasedFLSystemRuckel} requires at least $3 \revealTXDelay$ blocks, as computers must commit to the data-set to be used, before submitting a computation result. All 3 protocols, including our own, are equipped to spawn arbitrarily many computations in parallel.

\subsection{Necessary bounds on the number of computers per-computation}

\begin{table}[H]
\begin{center}
    
\scalebox{0.83}{
  \begin{tabular}{ |r|r|r|r|r|r|r|r|}
    \toprule
     &$\numComputers=$\textbf{10}&\textbf{25}&\textbf{50}&\textbf{100}&\textbf{250}&\textbf{500}&\textbf{1000}\\
     \midrule
    $\adversarialBound=$ \textbf{0.45}&0.504&0.694&0.716&0.817&0.936&0.986&0.999\\
    \midrule
    \textbf{0.4}&0.633&0.846&0.902&0.973&0.999&1&1\\
    \midrule
    \textbf{0.33}&0.787&0.958&0.989&1&1&1&1\\
    \midrule
    \textbf{0.25}&0.922&0.997&1&1&1&1&1\\
    \midrule
    \textbf{0.2}&0.967&1&1&1&1&1&1\\
    \midrule
    \textbf{0.1}&0.998&1&1&1&1&1&1\\
    \midrule
    \textbf{0.05}&1&1&1&1&1&1&1\\
    \midrule
    \textbf{0.01}&1&1&1&1&1&1&1\\
    \bottomrule
  \end{tabular}}
  \caption{The approximate probability, correct to 3 decimal places, of choosing a majority of rational computers given specific starting adversarial $\%$ of computers $\adversarialBound$ (left column) and selected numbers of computers $\numComputers$ (top row) from a sufficiently large population of computers such that adversarial share represents the per-selection probability of selecting an adversarial computer throughout sampling.} \label{table:advShares}
  
\end{center}
\end{table}

To estimate the required value of $\numComputersBound$ required to ensure a sufficiently probably majority of rational players, Table \ref{table:advShares} serves as a starting point. We recommend conservative choices of adversarial share and probability of rational majority. For example, given an adversarial share of $20\%$ of computers, any value of $\numComputersBound\geq 25$ suffices to ensure a majority of rational computers are chosen per-computation. 

\section{Conclusion}

We present $\protocolName$, a strong incentive compatible blockchain-based decentralized DC protocol which stands as a new standard in constructing fully decentralized DC protocols. This is achieved through a novel combination of strong incentivisation of rational computers in the presence of Byzantine computers, reputation-aware computer selection and privacy-enhancing techniques which can be bootstrapped to the core $\protocolName$ protocol to allow for the use of $\protocolName$ on computations revealing sensitive data about the computers participating in the protocol. Much work remains to ensure DC protocols remain incentive compatible and practical where computations produce large outputs, with storage being a limiting resources in mainstream blockchain protocols. $\protocolName$ and Privacy $\protocolName$ serve as key protocols with which to continue this research. 

\section*{Declaration of Competing Interest}

The authors declare that they have no known competing financial interests or personal relationships that could have appeared to influence the work reported in this paper.

\section*{Acknowledgements}

We thank IEEE Fellow Josep Domingo-Ferrer and his group at Universitat Rovira i Vigili for many interesting and motivating conversations related to this paper.

\section*{Vitaes}

\subsection*{Conor McMenamin}
\begin{figure}[H]
	\includegraphics[width=0.25\textwidth]{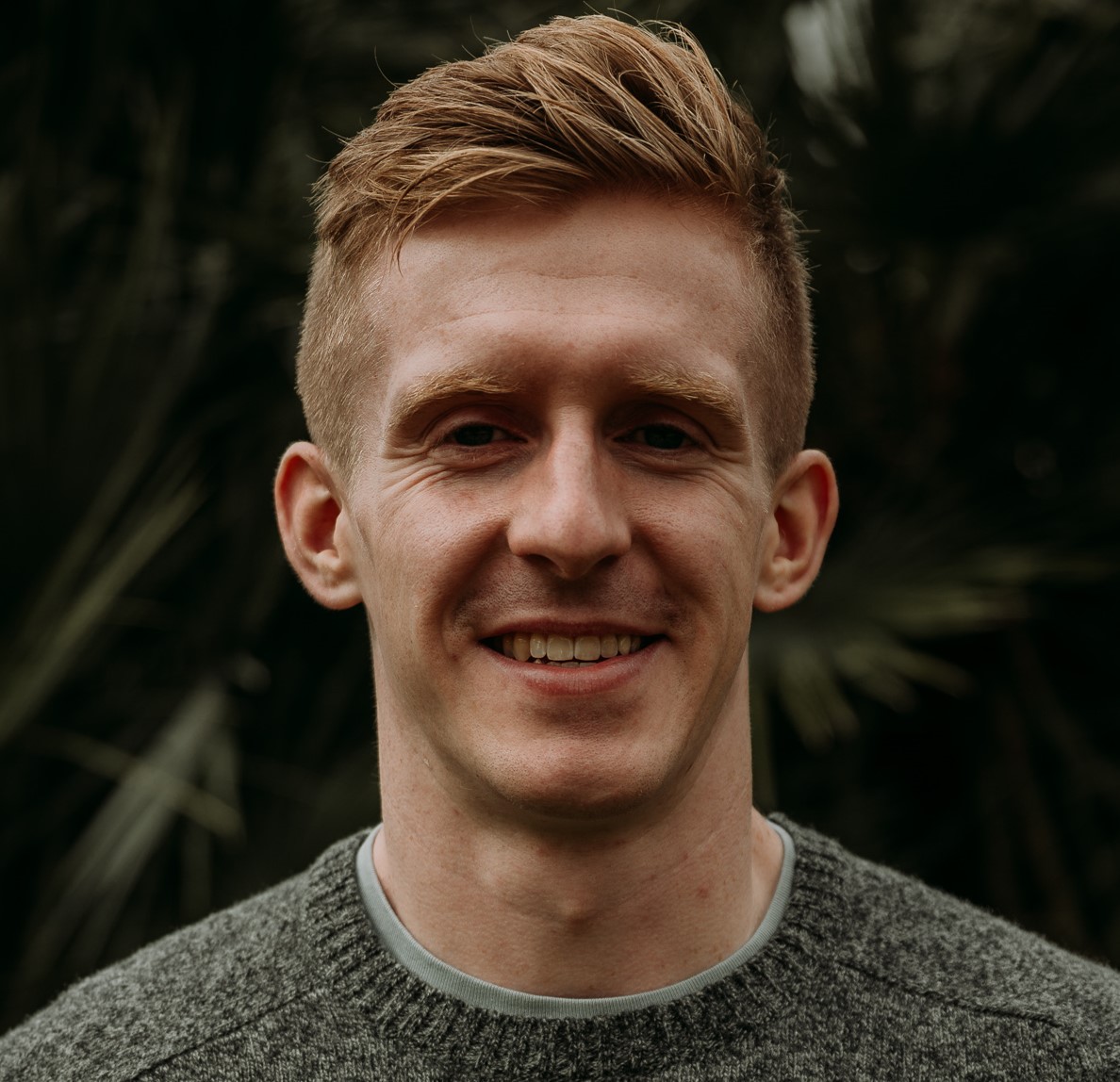}
\end{figure}

Conor is a Marie Curie Fellow, working as researcher at Nokia-Bell Labs undertaking a PhD in Universitat Pompeu Fabra (UPF). He holds a Bachelor of Science degree in Computational Thinking (2015), and a Master of Science degree in Cryptography (2019). His research focuses on the game-theoretic aspects of blockchain technology, both at the consensus and the application layer.

\subsection*{Vanesa Daza}
\begin{figure}[H]
	\includegraphics[width=0.3\textwidth]{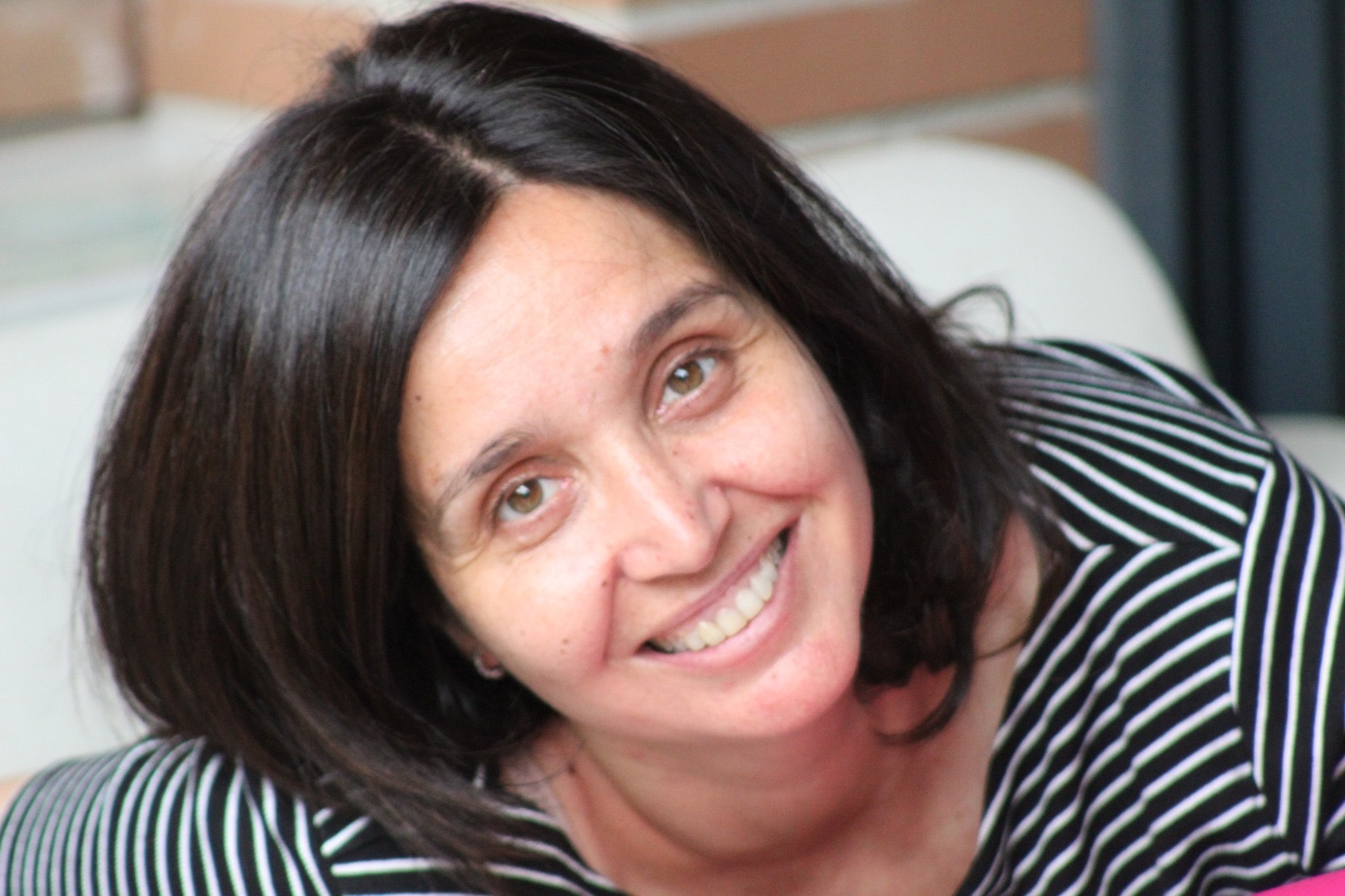}
\end{figure}

Vanesa is an Associate Professor at UPF. She holds a degree and a PhD in Mathematics.
Her research now focuses on designing cryptographic primitives and protocols to undertake  security and privacy challenges on blockchain technologies. She has co-authored over 40 papers and is the co-inventor of two international patents.
She has served as Associate Editor for two IEEE Transactions journals, several Program Committees and reviewer of international journals.
She currently coordinates two H2020 projects and collaborates with several companies. Likewise, she has served at UPF in several roles, including chairing the Department of Information and Communication Technologies.

\addcontentsline{toc}{section}{Bibliography}
\bibliographystyle{plain}
\bibliography{references}

\end{document}